\documentclass[11pt,a4paper,english]{article}

\pdfoutput=1
\usepackage[dvipdfmx]{color}
\usepackage[pdftex]{graphicx}

\usepackage[T1]{fontenc}
\usepackage[utf8]{inputenc}
\usepackage{babel}
\usepackage{blindtext}
\usepackage{times}
\usepackage{comment}
\def \fsize{10cm}
\def \fext{pdf}
\def \vect#1{\mbox{\boldmath $#1$}}

\newcommand{\argmax}{\mathop{\rm arg~max}\limits}

\newcommand{\rif}{{\rm ~~if~}}

\newcommand{\md}{\mathcal}

\usepackage{bm,enumerate,amsmath,amssymb,amsthm}
\usepackage{epsfig}
\usepackage{url}
\usepackage{algorithm}
\usepackage[noend]{algpseudocode}
\usepackage{cite}

\newtheorem{definition}{Definition}

\newtheorem{theorem}{Theorem}
\usepackage[title,titletoc]{appendix}

\algnewcommand\Input{\item[{\textbf{Input:}}]}
\algnewcommand\Output{\item[{\textbf{Output:}}]}
\algrenewcommand\algorithmicdo{}
\algnewcommand\algorithmicto{\textbf{to}}
\algdef{SE}[FOR]{ForTo}{EndFor}[2]{\algorithmicfor\ #1\ \algorithmicto\ #2\ \algorithmicdo}{\algorithmicend\ \algorithmicfor}%
\algtext*{EndFor}
\let\oldReturn\Return
\renewcommand{\Return}{\State\oldReturn}

\excludecomment{hdn}
\excludecomment{jap}
{ \begin{hdn} \begin{verbatim} }%
		{ \end{verbatim} \end{hdn} }

\title{Linear Game Theory : Reduction of complexity by decomposing large games into partial games}
\author{%
  Tatsuya Iwase \\
  \small Social Systems Research-Domain\\
  \small Toyota Central R\&D Labs., Inc. \\
  \small\texttt{tiwase@mosk.tytlabs.co.jp}
  \and
  Takahiro Shiga\\
  \small Social Systems Research-Domain\\
  \small Toyota Central R\&D Labs., Inc. \\
  \small\texttt{t-shiga@mosk.tytlabs.co.jp}
}

\date{}

\begin{document}
  \maketitle

\begin{abstract}
	With increasing game size, a problem of computational complexity arises. This is especially true in real world problems such as in social systems, where there is a significant population of players involved in the game, and the complexity problem is critical. Previous studies in algorithmic game theory propose succinct games that enable small descriptions of payoff matrices and reduction of complexities. However, some of the suggested compromises lose generality with strict assumptions such as symmetries in utility functions and cannot be applied to the full range of real world problems that may be presented. Graphical games are relatively promising, with a good balance between complexity and generality. However, they assume a given graph structure of players' interactions and cannot be applied to games without such known graphs. This study proposes a method to identify an interaction graph between players and subsequently decompose games into smaller components by cutting out weak interactions for the purpose of reducing complexity. At the beginning, players' mutual dependencies on their utilities are quantified as variance-covariance matrices among players. Then, the interaction graphs among players are identified by solving eigenvalue problems. Players' interactions are further decomposed into linear combinations of games. This helps to find a consistent equilibrium, which is a Nash equilibrium specified by the decomposition, with reduced computational complexity. Finally, experiments on simple example games are shown to verify the proposed method.
\end{abstract}
\newpage

\begin{hdn}
	\color{red}
	\begin{verbatim}
	【claim】利得変化をみることで自動的にggの構造を推定・分解できる。もしうまい分解が与えられれば、NE計算早くなる
	うまくかけてない
	・均衡を変えない分解である…？
	・succinctの参考文献追加
	理論的に不明
	・一貫均衡ってすでにある精緻化概念のサブクラスだったりしない？？：分解するってのは対戦相手の順番を変えるってことでは？？つまりサブゲー完全？サブゲー(部分ゲー)という概念自体、展開型の話だから、ちょっと違う。プレイヤ数変わんないし
	・一貫均衡の定義は？ある分解について？それとも分解が存在するもの？
	・分解できないトレードオフ均衡ってそんざいする？
	・vkはグループ？1プレイヤ？各vkに複数プレイヤが入ってる場合の意味。計算上は複数プレイヤ
	\end{verbatim}
	\color{black}
\end{hdn}

\section{\label{sec_intro}Introduction}

\begin{hdn}
	\color{red}
	\begin{verbatim}
	#G:社会シスはゲー理論で記述できる
	#社会シス同士も相互作用
	#しかも人口多い。大規模げーだ
	#現実の社会現象をげー理論で解決したい。現実都市・交通問題は大規模。ふつう解けない
	#予測→メカデザ・シグナリングで最適へ誘導
	#シェアげー軽く引用
	#何を計算するか。解の種類と計算量
	#P:大規模げーは計算遅い
	\end{verbatim}
	\color{black}
\end{hdn}

Lately, the impact of sharing economies such as Uber and Airbnb have started to replace the traditional economic systems. Consequently, a new effective technology is required for designing social systems such as cities and transportation. Since social systems consist of interactions between people and can be modeled as games between them, game theory can be a fundamental tool to analyze and optimize those systems. For instance, we modeled ride sharing services as a game among people who share vehicles, and studied a coordination mechanism to realize effective use of shared vehicles\cite{rsgame}.


However, when the size of games increases, computational complexity also increases. Generally speaking, a solution of a game such as Nash equilibrium or social optimum is denoted as a certain strategy profile that is a combination of all players' choices. Then, the computational complexity of a game depends on the size of the search space of solutions which increases exponentially according to the number of players. Especially in games of social systems, the population of involved players is usually huge, and the complexity problem is critical.

This study proposes a technique that reduces the size of games and accordingly their computational complexity, by decomposing large games into smaller components.

\subsection{\label{sec_past}Related work}

\begin{hdn}
	\color{red}
	\begin{verbatim}
	%B:高速化手法。あるごげー。CG
	%２人とN人
	%succinct game
	%n人：succinct:利得行列を簡素に。ふつうはna^n。むしろsuccinctでなければ解けない。
	%独立・対称
	%課題
	%一般性
	polyは2人でないとダメ
	%symmet
	%CGはシェア・内政変数ダメ
	%利得行列はsuccinct使えば簡単になる。計算量は別
	%gg構造わからん
	%P:構造はgiven。大規模だと手におえない
	%相互作用はスパースなのでは。分割できれば…
	%ゲーム表現を一般的かつコンパクトにしたい。トレードオフがある！
	\end{verbatim}
	\color{black}
\end{hdn}

The complexity problem in game theory has been studied for years, and the accumulation of knowledge is systematized as one of the main topics of ``Algorithmic Game Theory"\cite{algogt}. Roughly speaking, the computational complexity theory in algorithmic game theory has two subfields, namely the theory of 2 players and the theory of $N(>2)$ players. The latter is the focus of this study. In games of $N$ players, where the number of choices of each player is $A$, the size of the payoff matrix is $\md{O}(NA^N)$. Hence, it is difficult to give a full description of such a payoff matrix in large games that have big $N$s. Thus, the computational complexity theory only focuses on succinct games that can avoid full descriptions of payoff matrices and allow brief representations of large games. The following are major succinct games:

\begin{itemize}
	\item Games exploiting independencies between players
	\begin{itemize}
		\item Graphical games. These are games that have graphical descriptions of dependencies between players' utilities\cite{gg}. These allow a full description of all players' utilities with only payoff matrices among players that have mutual interactions.
		\item Polymatrix games. These are games that decompose the description of $N$-player utilities into the sums of utilities of 2-player games\cite{poly}. 
		\item Sparse games. These are games with most of the elements of the payoff matrices having zero values\cite{sparse}.
	\end{itemize}
	\item Games exploiting symmetries among players
	\begin{itemize}
		\item Symmetric games. These are games where all players have the same utility function\cite{nash}.
		\item Congestion games. These are games where the utility of a player's choice depends only on the number of players selecting that choice\cite{cg}.
		\item Local effect games. This is a generalization of congestion games. The utility of a player depends on the number of players with the same choice and also on the number of players with other choices\cite{leg}.
	\end{itemize}
\end{itemize}

These succinct games allow compact representation of games. Besides, some of the games provide efficient algorithms that calculate equilibria. Meanwhile, succinct games have the drawback of spoiling generality. Succinct games cannot describe many real world problems. For example, only limited problems can satisfy the assumption of symmetric games that all players have the same utility function. Congestion games can model congestions, which are general phenomena in the real world. However, they cannot model complex problems where supplies of resources can change according to players' choices such as car-sharing or ride-sharing\cite{rsgame}.

In addition, simplifying the description of payoff matrices by succinct games does not mean reduction of computational complexity in calculating game solutions. For example, congestion games provide compact descriptions of utility functions which only depend on the number of players. However, if $N$ players have a common subset of choices of size $A$, then the size of the solution space becomes at least $\md{O}(A^N)$ because all of the $N$ players interact with each other. Here, the problem of computational complexity in equilibrium calculation is evident.

It is known that the complexity of calculating Nash equilibrium is PPAD-complete, which means that there is no known general polynomial-time algorithm\cite{ppad}. However, previous works have been trying to propose relatively efficient algorithms. The Lemke-Howson algorithm formulates the calculation of Nash equilibrium of two players as a linear complementary problem\cite{lemke}. In general cases of more than two players, the problem becomes multilinear and then more complex\cite{wilson,datta}. In the case of polymatrix games, the problem of more than two players maintains linearity, but diminishes generality. They are all non-polynomial problems, which then need techniques to tackle large games.

Some works propose scalable algorithms for equilibrium computation by utilizing specific structures of games. Abstraction of games is a technique to reduce the size of the games and their complexity of equilibrium computation\cite{abst}. The original large game is converted into a coarse smaller game by grouping some elements of games such as types and actions of the players or stages of the games. However, an assumption on partial ordering of player types must be satisfied to guarantee the equivalence of equilibria between original and abstracted games. In case of games with temporal structures, the sequence form representation is applied to calculate equilibrium of two players efficiently by formulating linear programs or linear complementary problems\cite{seq}. While the sequence form assumes perfect recall of players, a modeling technique of player strategies with Markov Decision Processes (MDPs) enables the relaxation of the assumption\cite{branis,mdp}. Its temporal scalability may possibly reduce the complexity further in combination with the population scalability of our proposed method.

In the meantime, graphical games keep a good balance between reduction of complexity and generality of the model\cite{gg} and can possibly be applied to many real world problems. Graphical games, as the name suggests, describe the structure of a game in a directed graph. Nodes of a graph represent players, and the node of one player has an edge connected to another player when the utility of the destination player is dependent on the choice of the origin player. In this paper, we call the graph an interaction graph of the game. Therefore, the game can be described with small payoff matrices including only neighbor players instead of the full size payoff matrix of $N$ players (Figure \ref{fig:gg}). Directions of edges are omitted in figures if both directions exist. Graphical games also provide algorithms for equilibrium calculations using similar algorithms as information propagation in graphs which are well studied in the machine learning field. However, graphical games assume a given graph structure and cannot be applied to large complex games if the interactions of players are unknown.

Large complex games are intractable in both representation of payoff matrices and computational complexity of solutions. However, many real world problems do not involve interactions of all players, but rather a set of interactions between a small group of players, as in graphical games. If the interactions of players can be detected automatically and large complex games can be decomposed into a set of small games, the representation of payoff matrices can be simple, and computational complexity can be reduced.

\begin{figure}[ht]
	\centering
	\includegraphics[clip,width=\fsize]{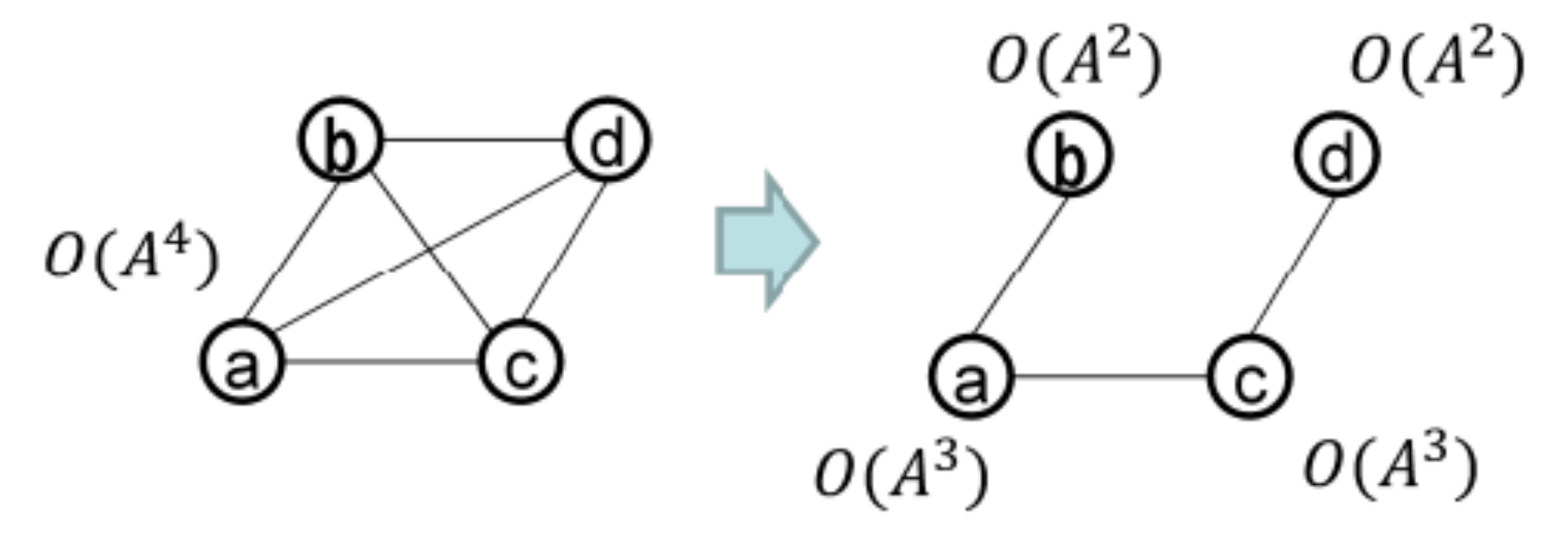}
	\caption{Graphical game}
	\label{fig:gg}	
\end{figure}

\subsection{\label{sec_goal}Purpose of the study}

\begin{hdn}
	\color{red}
	\begin{verbatim}
	A:仮説。相互作用はスパース、な場合もある→はじめに
	%でかいのそのまま解くのは無理→独立な部分に分割してから最適化しようぜ
	\end{verbatim}
	\color{black}
\end{hdn}

The purpose of this work is to decompose large complex games into sets of small games for reduction of computational complexity. We propose a method that covers a wide range of game classes.

\section{\label{sec_approach}Approach}

\subsection{\label{sec_problem}Problem of game decomposition}

\begin{hdn}
	\color{red}
	\begin{verbatim}
	%ゲーム分割、とは？
	%入力
	%あるプレイヤのu(a)を分解する、とは＝影響する小さいaを見つけること
	%分割ゲームと結合ゲーム定義
	%uk(a)=uk(ak)
	%ak/nkの定義：ｇｇの表記をまねる
	\end{verbatim}
	\color{black}
\end{hdn}

In this section, we formulate our problem of decomposing games. Our basic assumption is that not all players interact with each other, and the utility of a player depends on choices of only a limited number of players. Generally, each player has his own utility function and the set of players who have influence on the utility function is different for each player. Therefore, decomposition of a game means identifying the limited number of players who have influence on the utility function of each player.

A problem has an input game $\md{G}=<\md{N},\md{A},\vect{u}>$ to be decomposed. $\md{N}=\{1,...,N\}$ is a set of players, and $\md{A}_{i}$ is the set of choices of player $i \in \md{N}$. $a_{i} \in \md{A}_{i}$ is a choice of player $i$. $\md{A}=\underset{i \in \md{N}}{\times}\md{A}_{i}$ is a set of combinations of all players' choices. $\vect{a}=(a_{1},...,a_{N}) \in \md{A}$ is a combination of all players' choices called a strategy profile. $\vect{a}_{-i}$ denotes a strategy profile of all players except for $i$.

$u_{i}:\md{A} \to \mathbb{R}$ is the utility function of player $i$ and $\vect{u}=(u_{1},...,u_{N})$ is the utility vector function of all players. We use $u_{i}(a_{i}, \vect{a}_{-i})$ to denote $u_{i}(\vect{a})$ when it is more convenient. The shape of $\vect{u}$ is arbitrary, but its representation must be succinct because the input game $\md{G}$ is large. We assume that it is possible to calculate the utility of a player $\vect{u}(\vect{a})$ with a given profile $\vect{a}$, but interaction graphs are not given as in graphical games. Hence, a search of the whole profile space $\md{A}$ is necessary to calculate a solution of $\md{G}$, and it is computationally intractable. This is because the utility function $u_{i}(\vect{a})$ takes a strategy profile of all players as input. Instead, we should identify a limited number of players who have influence on $u_{i}$ to reduce complexity.

\begin{definition}[Partial profile]
	Given a subset of players $\md{N}' \subset \md{N}$, $\vect{a}_{\md{N}'}=(...,a_{(\md{N}',i)},...) \in \md{A}_{\md{N}'}$ is a partial profile of $\vect{a} \in \md{A}$ iff $\md{A}_{\md{N}'}=\underset{i \in \md{N}'}{\times}\md{A}_{i}$ and $a_{(\md{N}',i)}=a_{i},\forall i \in \md{N}'$. This is denoted as $\vect{a}_{\md{N}'} \subset \vect{a}$ for simplicity.
	\label{def:partiala}
\end{definition}
\begin{definition}[Influencer]
	Player $j$ is an influencer of player $i$ iff $u_{i}$ depends on $j$'s choice. The set of influencers for player $i$ is denoted as $\md{N}^i$. We denote as follows for simplicity:
	\begin{equation}
		u_{i}(\vect{a})=u_{i}(\vect{a}_{\md{N}^i}).
	\end{equation}
	\label{def:influencer}
\end{definition}
\begin{definition}[Partial game]
	Given a subset of players $\md{N}_k \subset \md{N}$, $\md{G}_k=<\md{N}_k,\md{A}_{\md{N}_k},\vect{u}_{k}>$ is a partial game of $\md{G}=<\md{N},\md{A},\vect{u}>$ iff
	\begin{equation}
		\left.
		\begin{array}{l}
			\vect{u}_{k}=(...,u_{(k,i)},...)_{i \in \md{N}_k}, \\
			u_{(k,i)}(\vect{a}_{\md{N}_k})=u_i(\vect{a}_{\md{N}_k}), \forall i \in \md{N}_k, \forall \vect{a}_{\md{N}_k} \in \md{A}_{\md{N}_k}.
		\end{array}
		\right.
	\end{equation}
	\label{def:partg}
\end{definition}
Now decomposition of game $\md{G}$ is defined as follows.
\begin{definition}[Decomposed games]
	Partial games $\{\md{G}_k=<\md{N}_k,\md{A}_{\md{N}_k},\vect{u}_{k}>|k \in \md{K}=\{1,...,K\}\}$ are a decomposition of $\md{G}=<\md{N},\md{A},\vect{u}>$ iff
	\begin{equation}
		\md{N}=\underset{k \in \md{K}}{\cup}\md{N}_{k}.
	\end{equation}
	\label{def:decomg}
\end{definition}
In reverse, the game $\md{G}$ is the combined game of $\{\md{G}_k\}$.

The goal of the game decomposition problem is to identify influencers $\md{N}^{i}$ for all players and subsequently decompose the input game $\md{G}$ if possible. How can this be solved?

\subsection{\label{sec_overview}Overview of approach}

\begin{hdn}
	\color{red}
	\begin{verbatim}
	%アプローチ全体概要：不明なのを、相互作用検出・構造推定して、分割
	%二段階の分割、二種類の切断・分裂
	%相対的に弱いのを切る
	claim
	\end{verbatim}
	\color{black}
\end{hdn}

Our approach has two steps. The first step is finding interactions among players and estimating influencers $\md{N}^{i}$ for all players. This is equivalent to estimating the interaction graph of the game. Our proposed method achieves this first step by constructing variance-covariance matrices among players. Once the matrix of player $i$ is obtained, the influencers $\md{N}^{i}$ are estimated by solving the eigenvalue problem.

After identifying the graph of the game, the second step further decomposes them if possible. This is also done two ways. One way is approximate cutting. Since our approach can reveal the relative strength of interactions between players, weak interactions can be cut to make smaller games. The other way is division of a player. This way decomposes the utility function $u_{i}$ into a linear combination of two small utility functions of player $i$. The following sections describe this in further detail.

\subsection{\label{sec_principal}Principle of detecting interactions between players}

\begin{hdn}
	\color{red}
	\begin{verbatim}
	%差分をとる
	\end{verbatim}
	\color{black}
\end{hdn}

Figure \ref{fig:principle} shows simple 2-player games and describes the principle of how our approach detects interactions between players. There are two payoff matrices of 2 players who have choices of $\{a_{R},a_{L}\}$. The left one shows a case where 2 players have interactions with each other, and the right one shows a case without any interactions between players. In the left case, the payoff of the row player changes according to the choice of the column player. Meanwhile in the right case, the payoff of the row player does not depend on the choice of the column player. Therefore, as shown in the bottom of the figure, the change in payoff when the opponent changes the choice represents the strength of the influence of the opponent's actions on the payoff.

\begin{figure}[ht]
	\centering
	\includegraphics[clip,width=\fsize]{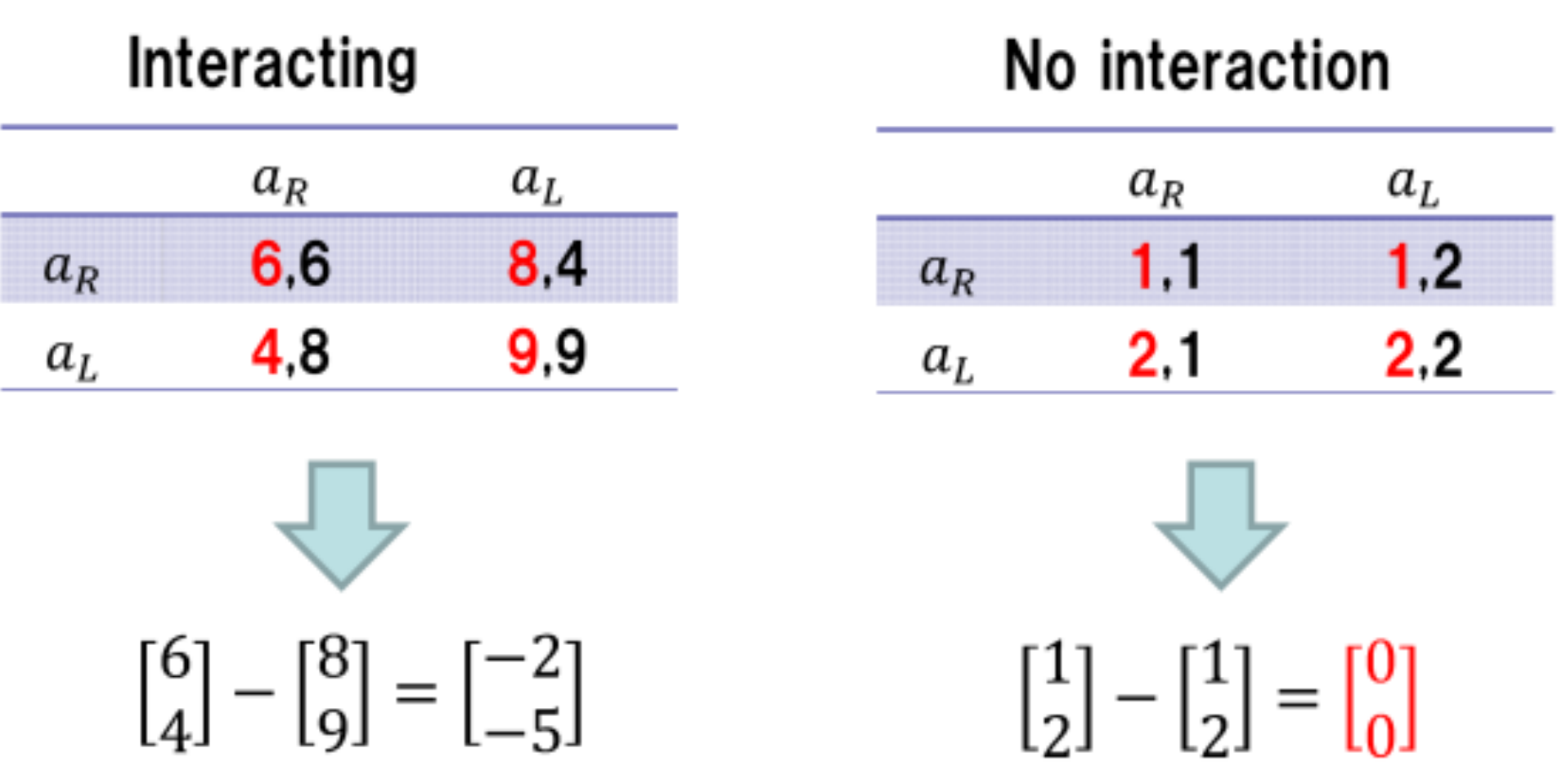}
	\caption{Principle of detecting interactions}
	\label{fig:principle}	
\end{figure}

\subsection{\label{sec_cov}Interaction matrix between players}

\begin{hdn}
	\color{red}
	\begin{verbatim}
	%共分散の作り方
	\end{verbatim}
	\color{black}
\end{hdn}

In this section, we describe how to detect interactions between $N$ players, whereas the last section describes 2-player cases. The proposed method makes concrete the interactions between $N$ players on the utility of player $i$ into a variance-covariance matrix $C_{i} \in \mathbb{R}^{N \times N}$ that we call the interaction matrix of player $i$.

The procedure \textsc{InteractionMatrix} given in Algorithm \ref{alg:imat} computes the matrix. A game instance $\md{G}=<\md{N},\md{A},\vect{u}>$ is given as the input, but the interactions among players are unknown. A difference of payoff $\Delta u_{i}$ is calculated by changing the choice of player $j$ between two strategy profiles $\vect{a},\vect{b}$ (line 9), as illustrated in Figure \ref{fig:sampling}. The procedure repeats this step to collect samples of all players. Then the output $C_{i}$ has the information of the interactions between players on the utility of player $i$. The size of data needed to calculate $C_{i}$ for all players is at most $\md{O}(LN^2)$.

\begin{algorithm}                      
	\caption{Interaction matrix}         
	\label{alg:imat}                          
	\begin{algorithmic}[1]
		\Input A game instance $\md{G}=<\md{N},\md{A},\vect{u}>$; A player index $i$; Number of samples for each player $L$
		\Output An interaction matrix $C_{i}$
		\Procedure{InteractionMatrix}{$\md{G},i$}
		\State{Initialize $X_{i} \in \mathbb{R}^{NL \times N}$ as a zero matrix.}
		\State{$k=0.$}
		\For{$j \in \md{N}$}
		\ForTo{$l=1$}{$L$}
		\State{Choose $\vect{a} \in \md{A}$ randomly.}
		\State{$\vect{b} = \vect{a}.$}
		\State{Update $b_{j} \in \md{A}_{j}\setminus\{a_{j}\}$} randomly.
		\State{$\Delta u_{i}=u_{i}(\vect{b})-u_{i}(\vect{a}).$}
		\State{$X_{i}[k][j]=\Delta u_{i}.$}
		\State{$k=k+1.$}
		\EndFor
		\EndFor
		\State{Compute $\bar{X_{i}} \in \mathbb{R}^{NL \times N}$ with identical rows which are the average of $X_{i}$ along rows.}
		\State{$C_{i}=(X_{i}-\bar{X_{i}})^T(X_{i}-\bar{X_{i}})/(NL-1)$.}
		\Return{$C_{i}$}
		\EndProcedure
	\end{algorithmic}
\end{algorithm}

\begin{figure}[ht]
	\centering
	\includegraphics[clip,width=\fsize]{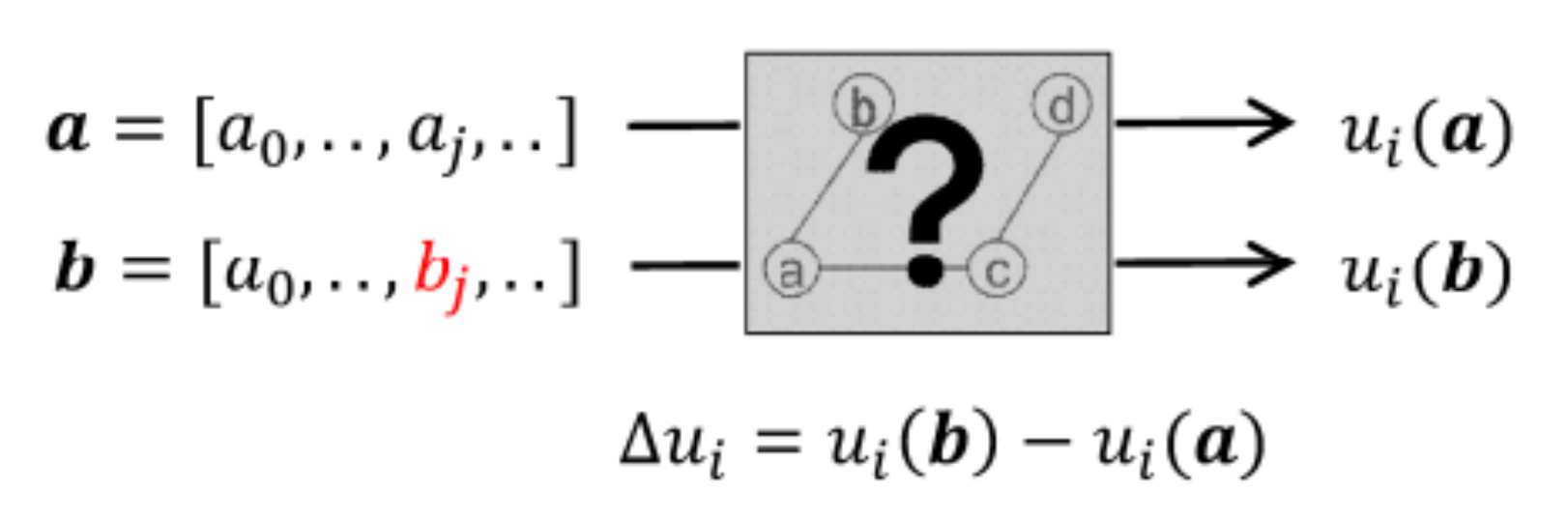}
	\caption{Sampling data of utility differences}
	\label{fig:sampling}	
\end{figure}

\subsection{\label{sec_influ}Estimation of influencers}

\begin{hdn}
	\color{red}
	\begin{verbatim}
	%1.グラフの推定
	%lambda>0なのがinfluencer
	%各プレイヤiごとの話
	%固有ベクトルP：利得発生の元になったある力・相互作用(プレイヤグループ)
	%固有値L：そのプレイヤグループの影響力
	%固有値>0のプレイヤグループでグラフを再構成
	\end{verbatim}
	\color{black}
\end{hdn}

The proposed method identifies influencers $\md{N}^{i}$ in Definition \ref{def:influencer} out of all players $\md{N}$ based on the interaction matrix $C_{i}$ and linear algebra. At first, all players $\md{N}$ are clustered into groups $k \in \md{K}^i =\{1,...,K_{i}\}$ by solving the following eigenvalue problem.

\begin{equation}
	C_{i}\vect{v}_{k}=\lambda_{k}\vect{v}_{k}.
	\label{eq:eigen}
\end{equation}

We consider only nontrivial solutions that have $\lambda_{k} > 0$. An eigenvector $\vect{v}_{k} \in [-1,1]^{N}$ represents a group of players who have influence on the utility of player $i$. If player $j$ is an influencer of player $i$, the element $j$ of $\vect{v}_{k}$ has a non-zero value. The elements of players who are not influencers have all zero values. An eigenvalue $\lambda_{k}$ represents the relative strength of influence of players in $\vect{v}_{k}$. Thus, all influencers $\md{N}^{i}$ are identified by checking $\lambda_{k}$ and $\vect{v}_{k}$. Once $\md{N}^{i}$ are identified for all players $i \in \md{N}$, they contain enough information to draw the interaction graph of the game. The computational complexity is now reduced because the calculation of utility $u_{i}(\vect{a}^i)$ depends only on the profile of influencers $\vect{a}^i$.

\subsection{\label{sec_aprox}Approximate cut of game}

\begin{hdn}
	\color{red}
	\begin{verbatim}
	%2.新ゲームを構成
	%固有値小さいやつ近似もできる
	\end{verbatim}
	\color{black}
\end{hdn}

The eigenvalues and eigenvectors computed in the last section contain more information than the interaction graph of the game by itself. Since they have information on the relative strengths of interactions between players, weak interactions in the game can be identified by setting thresholds of eigenvalues and eigenvectors. Hence, it is possible to decompose the game approximately by cutting out the weak connection in the interaction graph (Figure \ref{fig:cutg}). An example is shown in Section \ref{sec_ggdecom}.

\begin{figure}[ht]
	\centering
	\includegraphics[clip,width=6cm]{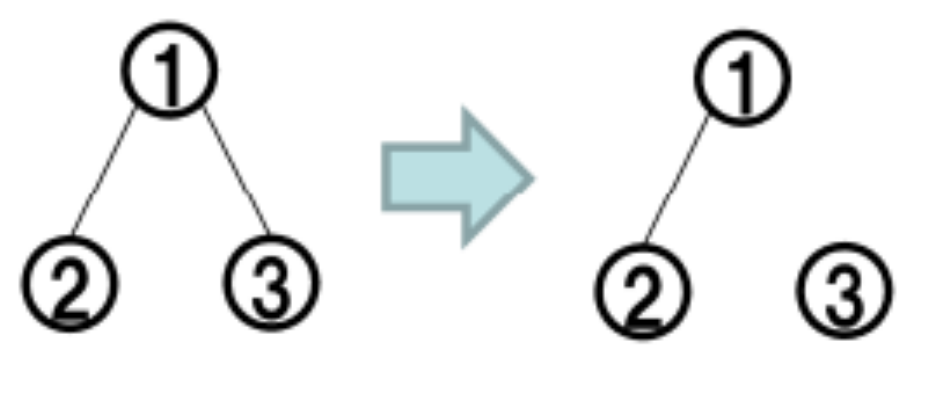}
	\caption{Cut of a game}
	\label{fig:cutg}	
\end{figure}

\subsection{\label{sec_lg}Division of players in linear combinations of games}

\begin{hdn}
	\color{red}
	\begin{verbatim}
	%ゲームの線形結合定義
	%線形結合によるggの更なる分裂
	%一貫均衡
	トレードオフ均衡
	%NEの求め方
	\end{verbatim}
	\color{black}
\end{hdn}

Figure \ref{fig:divp} shows another decomposition of the same game in Figure \ref{fig:cutg}. Different from Figure \ref{fig:cutg}, player 1 belongs to both decomposed games. Even though it changes the interaction graph and properties of the game, this decomposition is useful for reducing certain types of complexities of the original game. This section describes a case of equilibrium computations.

\begin{figure}[ht]
	\centering
	\includegraphics[clip,width=7cm]{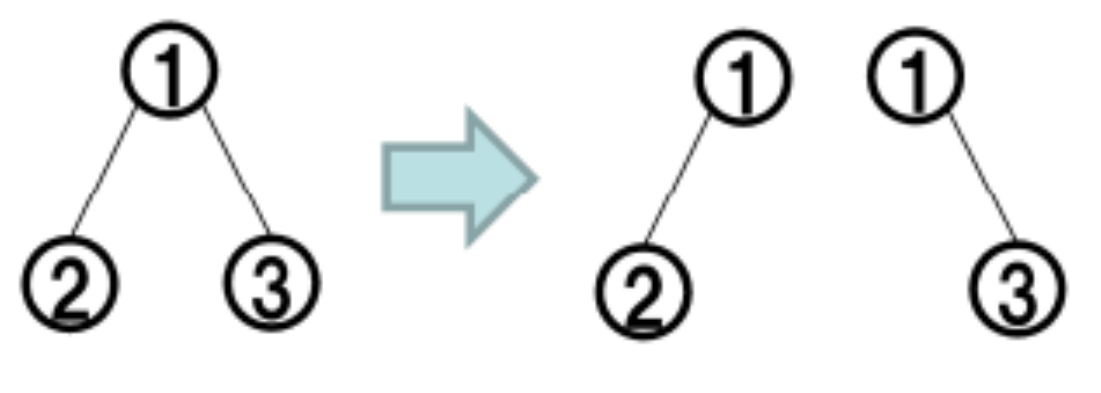}
	\caption{Division of a player}
	\label{fig:divp}	
\end{figure}

\begin{definition}[Nash satisfaction]
	Let $\md{G}_k=<\md{N}_k,\md{A}_{\md{N}_k},\vect{u}_{k}>$ denote a partial game of $\md{G}=<\md{N},\md{A},\vect{u}>$. A strategy profile $\vect{a} \in \md{A}$ satisfies $\md{G}_k$ iff $\vect{a}_{\md{N}_k}^* \subset \vect{a}$ is a Nash equilibrium of $\md{G}_k$.
	\label{def:ns}
\end{definition}

\begin{definition}[Linear decomposition of game]
	Let $\{\md{G}_k=<\md{N}_k,\md{A}_{\md{N}_k},\vect{u}_{k}>\}$ denote a decomposition of $\md{G}=<\md{N},\md{A},\vect{u}>$. $\{\md{G}_k\}$ is a linear decomposition of $\md{G}$ iff
	\begin{equation}
		\vect{u}(\vect{a})=\sum_{k \in \md{K}}b_{k}\vect{u}_{k}(\vect{a}_{\md{N}_k}),
		\label{eq:ldg}
	\end{equation}
	where $b_{k}$ is a constant. In reverse, $\md{G}$ is a linear combination of $\{\md{G}_k\}$.
	\label{def:ldg}
\end{definition}

The following theorem helps to find a Nash equilibrium of combined games by finding Nash equilibria of linearly decomposed games.

\begin{theorem}
	Let $\{\md{G}_k\}$ denote a linear decomposition of $\md{G}$ with $b_{k} \ge 0, \forall k \in \md{K}$. If a profile $\vect{a}$ satisfies all $\{\md{G}_k\}$, $\vect{a}$ is a Nash equilibrium of $\md{G}$.
	\label{thm:conse}
\end{theorem}

\begin{proof}
	If $\vect{a}=(a_{1},...,a_{N})$ satisfies all $\{\md{G}_k\}$, the following conditions of Nash equilibrium are satisfied:
	\begin{equation}
		a_{i} = \argmax_{a'_{i} \in \md{A}_{i}}u_{(k,i)}(a'_{i},\vect{a}_{(\md{N}_k,-i)}), \forall k \in \md{K}, \forall i \in \md{N}_k.
		\label{eq:pfconse}
	\end{equation}
	Since $\vect{a}$ maximizes $\vect{u}_{k}$ for all $k \in \md{K}$, $\vect{a}$ also maximizes their positively weighted sum $\vect{u}$ in \eqref{eq:ldg} as follows
	\begin{equation}
		a_{i} = \argmax_{a'_{i} \in \md{A}_{i}}u_i(a'_{i},\vect{a}_{-i}), \forall i \in \md{N}
		\label{eq:pfconse2}
	\end{equation}
	which means that $\vect{a}$ is a Nash equilibrium of $\md{G}$.
\end{proof}

Now a decomposition can specify a Nash equilibrium that is defined as follows:

\begin{definition}[Consistent equilibrium]
	Let $\{\md{G}_k\}$ denote a decomposition of $\md{G}$. A Nash equilibrium $\vect{a}^* \in \md{A}$ of $\md{G}$ is a consistent equilibrium of $\{\md{G}_k\}$ iff $\vect{a}^*$ satisfies all $\{\md{G}_k\}$.
	\label{def:conse}
\end{definition}

Hence, if we decompose a game linearly as \eqref{eq:ldg}, it can help to find a consistent equilibrium of the original combined game in accordance with Theorem \ref{thm:conse}. Since the number of players decreases in each decomposed game, the complexity required to find equilibria is also reduced. This approach is similar to the polymatrix games, but can cover a broader class of games. An example is shown in Section \ref{sec_lgdecom}.

\section{\label{sec_example}Examples}

\subsection{\label{sec_ggdecom}Example of identifying an interaction graph}

\begin{hdn}
	\color{red}
	\begin{verbatim}		
	%R:ブラックボックスげー。何を知っているか
	%分割できた
	%近似も可能
	\end{verbatim}
	\color{black}
\end{hdn}

In this section, the method proposed in Section \ref{sec_approach} is verified to see if it can identify the interaction graph of an input game $\md{G}=<\md{N},\md{A},\vect{u}>$. Figure \ref{fig:orig} shows the true interaction graph of $\md{G}$ which is unknown to the proposed method. The verification compared the estimation of the proposed method with the true graph. $\md{N}$ includes 24 players and a choice of each player $i \in \md{N}$ is binary $\md{A}_{i}=\{0,1\}$. The player $i$'s utility function $u_i(\vect{a})$ is generated by uniform random sampling of a value from integers $\{0,...,9\}$ for all partial profiles of influencers $\vect{a}_{\md{N}^i} \in \md{A}_{\md{N}^i}$. If the interaction graph is unknown, the complexity of equilibrium computation is $|\md{A}|=2^{24}$.

According to the method in Section \ref{sec_cov}, $L=10$ sample data are collected for each player $i$ and then the interaction matrix $C_{i}$ and eigenvectors $\vect{v}_{k}, k \in \md{K}_{i}=\{1,...,K_{i}\}$ are calculated. Subsequently, the interaction graph can be estimated by regarding player $j$ who has $|v_{k,j}|>0$ in any of the $K_{i}$ eigenvectors as an influencer of player $i$. The proposed method correctly estimated the true interaction graph as in Figure \ref{fig:orig}. Additionally for example, the interaction matrix of player 10 has three eigenvectors which correspond to player 0, 10 and 20 respectively, who are all influencers of player 10 (Figure \ref{fig:eigenv}).

The size of sample data required for this estimation is $10*24^2$ which is smaller than the complexity without estimation $2^{24}$. In addition, the input game is decomposed into two partial games including 12 players each which reduces the complexity of equilibrium computation into $2^{12}$ at most. Furthermore, the biggest groups of players interacting with each other are two groups that have 5 players including players 0 and 23 respectively, because of which an efficient computation can be expected by applying the algorithm of graphical games.

Moreover, the eigenvalues $\lambda_{k}$ and eigenvectors $\vect{v}_{k}$ include more detailed information on the relative strengths of the interactions. According to the method in Section \ref{sec_aprox}, this information can be used for further approximate decomposition of the game. Figure \ref{fig:aprox} shows the result of decomposition that is a graph including only strong interactions of $\lambda_{k}>0.29, |v_{k,j}|>0.05$. The arrows in the figure show the cut edges. As a result, the game is decomposed into 4 small games which realizes further improvement in computation efficiency.

\begin{figure}[ht]
	\centering
	\includegraphics[clip,width=\fsize]{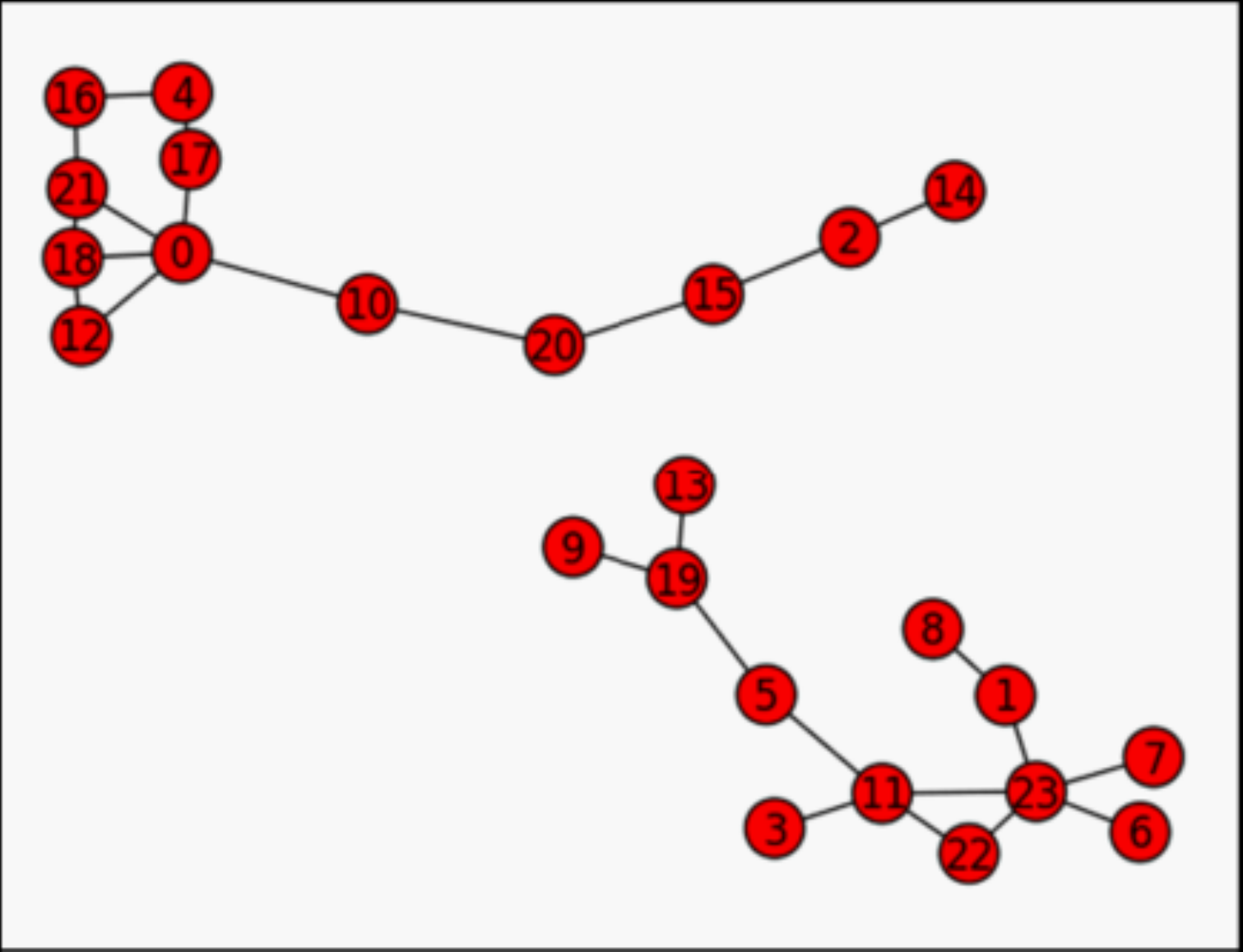}
	\caption{True interaction graph}
	\label{fig:orig}	
\end{figure}

\begin{figure}[ht]
	\centering
	\includegraphics[clip,width=6cm]{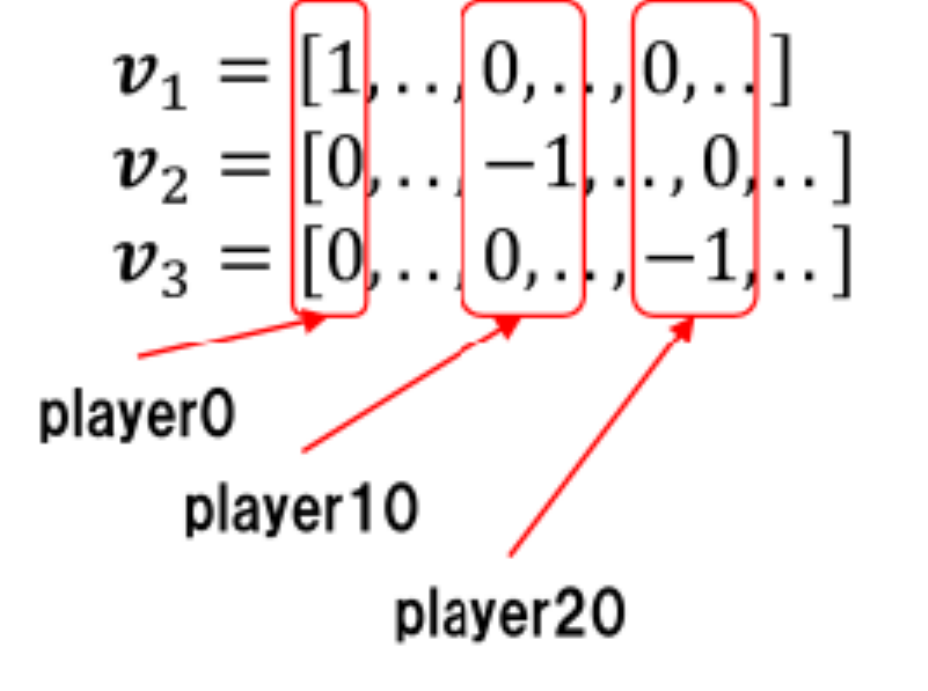}
	\caption{eigenvectors of player 10}
	\label{fig:eigenv}	
\end{figure}

\begin{figure}[ht]
	\centering
	\includegraphics[clip,width=\fsize]{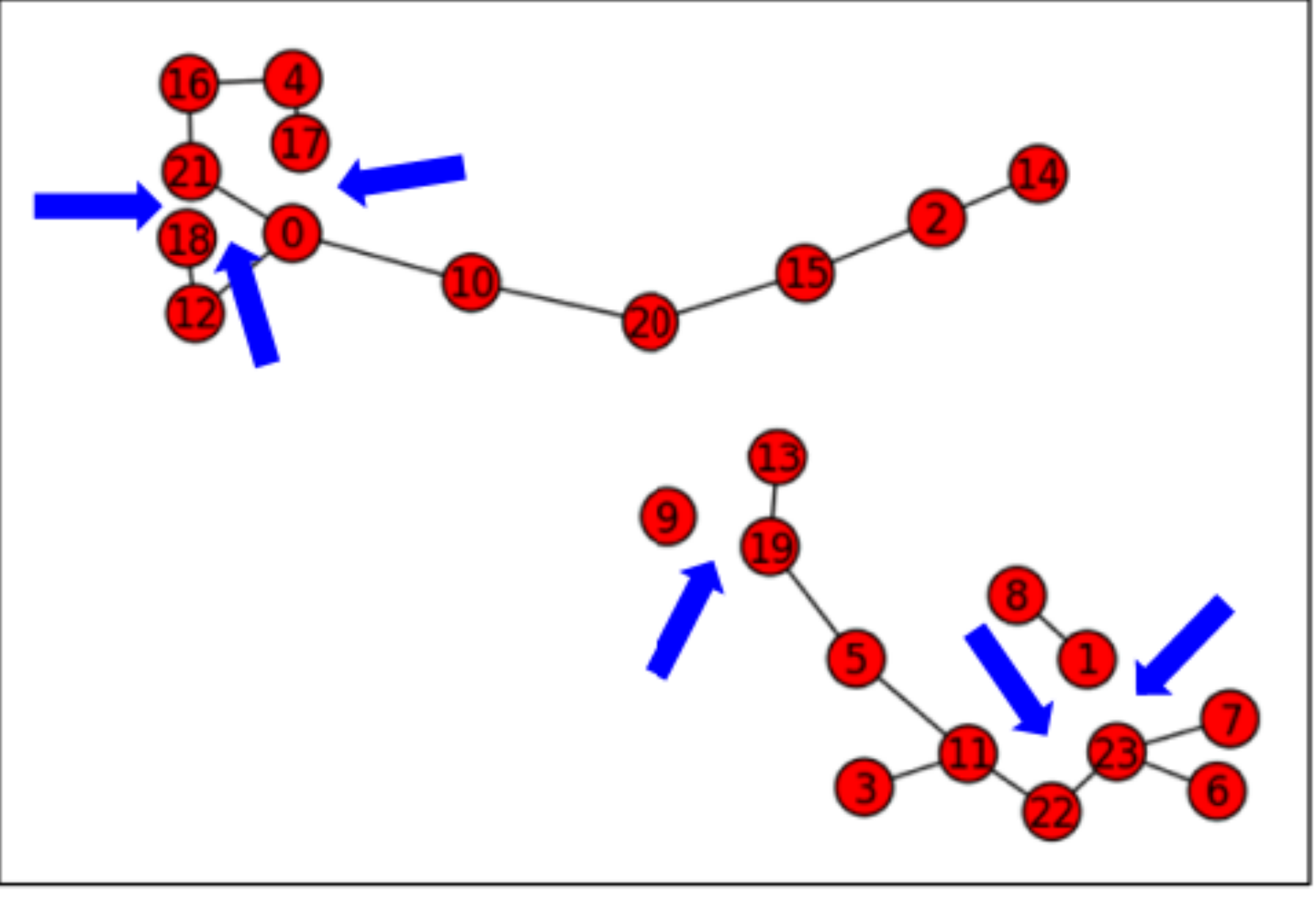}
	\caption{Approximate graph with only strong interactions}
	\label{fig:aprox}	
\end{figure}

\subsection{\label{sec_lgdecom}Example of linear decomposition}

\begin{hdn}
	\color{red}
	\begin{verbatim}		
	R:ブラックボックスげー。何を知っているか
	分割できた
	分割の自由度
	近似も可能
	ナッシュが簡単に求まる
	\end{verbatim}
	\color{black}
\end{hdn}

In this section, an example is shown to describe how a linear decomposition helps to find a consistent equilibrium. The input game $\md{G}=<\{1,2,3\},\md{A},\vect{u}>$ has a interaction graph similar to left one in Figure \ref{fig:divp}. The game is a linear combination of two partial games $\md{G}_2=<\{1,2\},\md{A}_2,\vect{u}_2>$ and $\md{G}_3=<\{1,3\},\md{A}_3,\vect{u}_3>$. The payoff matrices $\vect{u}_2$ and $\vect{u}_3$ are generated in a similar way as in the last section and are shown in Figure \ref{fig:uorg}. The utility function of $\md{G}$ is a linear combination $\vect{u}(\vect{a})=\vect{u}_2(\vect{a})+\vect{u}_3(\vect{a})$. Though $\vect{u}(\vect{a})$ is known input, this decomposition is unknown to the proposed method.

Now let's consider the problem of finding the Nash equilibria of $\md{G}$. Actually, this game has three Nash equilibria which are $\{(0, 2, 0),(0, 2, 1),(2, 0, 0)\}$. If the true decomposed payoff matrices in Figure \ref{fig:uorg} are unknown, finding those equilibria requires searching the whole space of strategy profile $\md{A}$, which has size $3^3=27$. However, the proposed method in Section \ref{sec_lg} enables finding an equilibrium with smaller complexity.

The method requires computing a linear decomposition of $\md{G}$ which is $\{\md{G}_2',\md{G}_3'\}$. This is actually computing a decomposition of utility $\vect{u}(\vect{a})=\vect{u}_2'(\vect{a})+\vect{u}_3'(\vect{a})$ without knowing the true decomposition in Figure \ref{fig:uorg}. The way to decompose the utility $\vect{u}(\vect{a})$ has a degree of freedom.

Figure \ref{fig:ugood} shows one of the possible decompositions. The Nash equilibria of the decomposed games $\{\md{G}_2',\md{G}_3'\}$ are $\{(0,2),(2,0)\}$ and $\{(2,0)\}$, respectively. In this case, a profile $(2,0,0)$ satisfies both decomposed games and then is a consistent equilibrium of the input game $\md{G}$ and is one of the Nash equilibria. Since the computation required to find this consistent equilibrium is calculating Nash equilibria of decomposed games, the complexity is $2*3^2=18$, which is smaller than the original complexity of $3^3=27$.

A consistent equilibrium requires a proper linear decomposition. Figure \ref{fig:ubad} shows another decomposition of the input game. The Nash equilibria of the decomposed games $\{\md{G}_2',\md{G}_3'\}$ are $\{(0,2)\}$ and $\{(2,0),(2,2)\}$ respectively. In this case, no profile of $\md{G}$ satisfies both decomposed games and thus it fails to find the consistent equilibrium.

\begin{figure}[h]
	\centering
	\includegraphics[clip,width=\fsize]{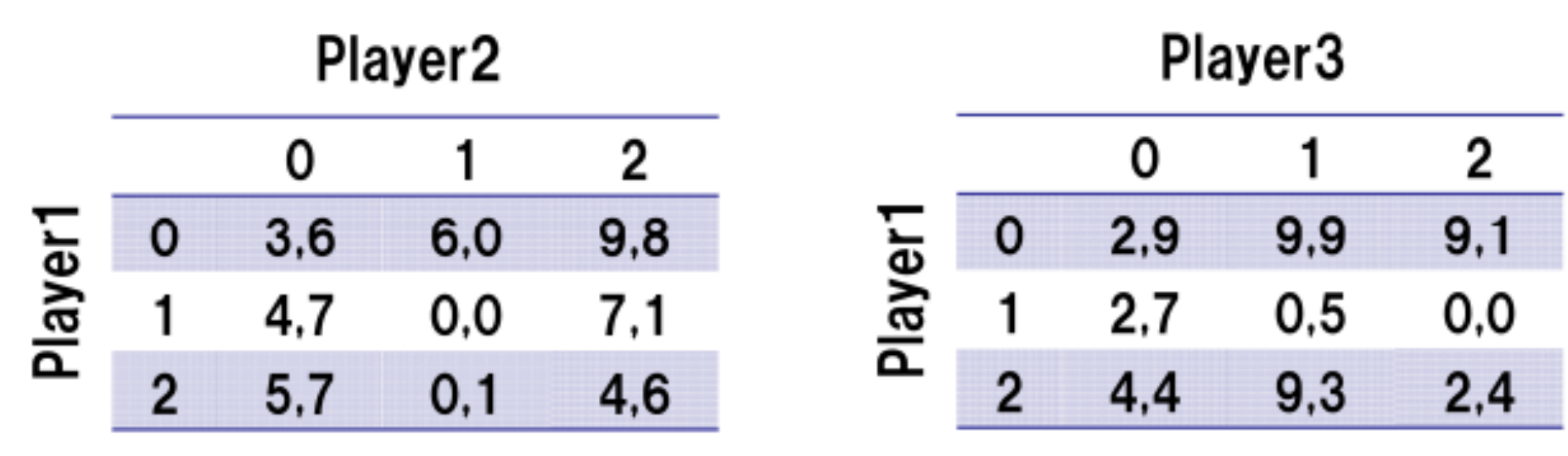}
	\caption{True payoff matrices of input games}
	\label{fig:uorg}	
\end{figure}

\begin{figure}[h]
	\centering
	\includegraphics[clip,width=\fsize]{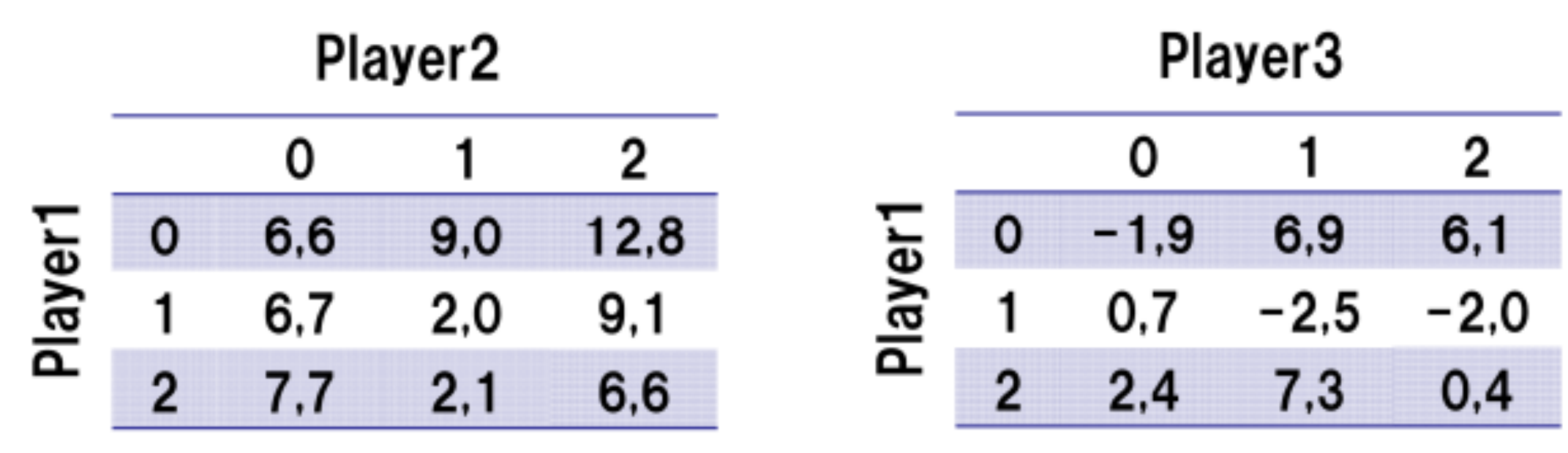}
	\caption{Decomposition with consistent equilibrium}
	\label{fig:ugood}	
\end{figure}

\begin{figure}[h]
	\centering
	\includegraphics[clip,width=\fsize]{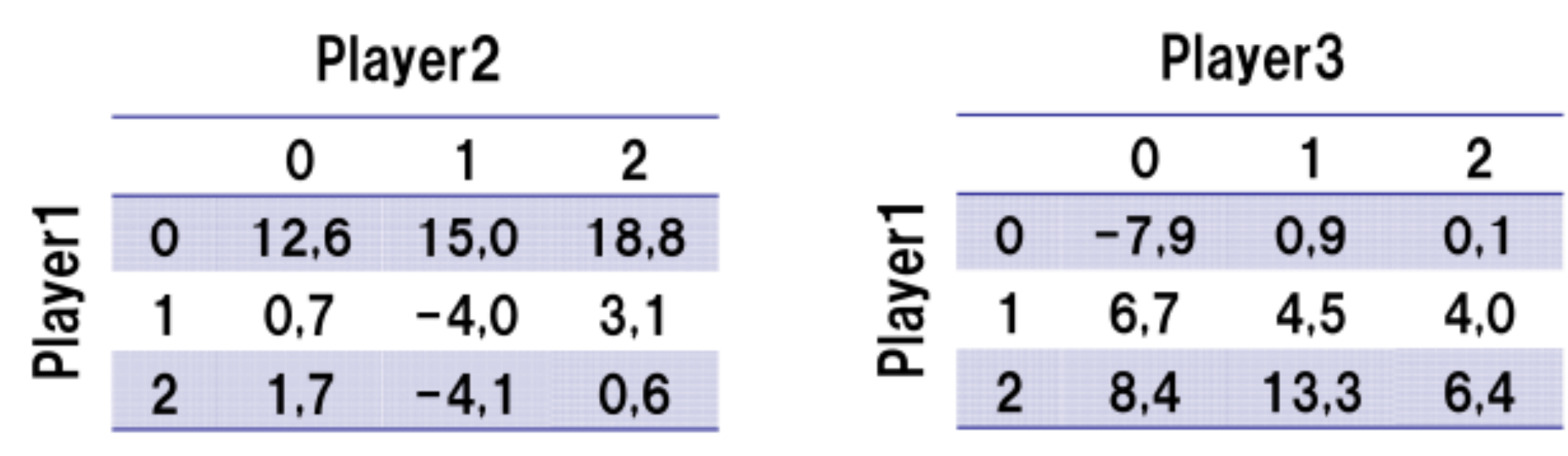}
	\caption{Decomposition without consistent equilibrium}
	\label{fig:ubad}	
\end{figure}


\section{\label{sec_conclude}Conclusion}

\begin{hdn}
	\color{red}
	\begin{verbatim}
	メリット
	#一般的。どんなゲームでも
	#社会システムは複数の人々のゲームとして記述できる場合が多い。
	#それらに汎用的に利用できる技術である
	#SO,ME,NE全部効く
	#ほかの高速化手段と組み合わせが可能
	応用先：近似、or ブラックボックスな場合。メカデザに使えるかも。複雑な内生変数を持つゲーム
	線形結合というsuccinctな表現を提案した
	
	#RA:シェアげーで確認
	#線形代数的にゲームを分割する研究はない。新規
	#線形のメリットは全て今後の話。…できる可能性がある
	#線形結合の意味
	#線形代数理論のフル活用
	線形結合ゲーの均衡
	スパースモデリングによる少量データでの独立検出
	戦略：分割した小規模ゲームの効率的解法
	利得行列と固有空間C=PLP^-1の相互変換定理の証明
	#ゲーム理論を固有空間法をはじめとする線形理論と接続できたことを意味する。線形理論の各ツールを利用できる可能性がある
	\end{verbatim}
	\color{black}
\end{hdn}

For the purpose of reducing computational complexity of large complex games, this study proposes a method to identify an interaction graph of the input game by constructing an interaction matrix between players and solving its eigenvalue problem. Additionally, a technique of linear decomposition of games is proposed to find consistent equilibria, which are Nash equilibria specified by the decomposition, with reduced computational complexity.

The proposed method is applicable to a wide range of game classes described with utility functions. Real world problems in the field of city planning, transportation, environment, and economics are generally modeled as interactions between people, so the applications of the proposed method are full of variety. Also, since the proposal is a simple method that decomposes large games into smaller ones, it can be combined with other methods that also reduce computational complexity.

Future studies include validation with real world problems and effective computation of each game decomposed by the proposed method. This study shows a method for connecting game theory with linear algebra. Since linear algebra has been applied to variety of areas and has many useful tools, the connection has the potential to evolve game theory further.


\begin{thebibliography}{99}
	
	\bibitem{rsgame}
	Tatsuya Iwase, Takahiro Shiga, (2016). "Coordination of Players in Ride-Sharing Games by Signaling," Proceedings of the 30th Annual Conference of the Japanese Society for Artificial Intelligence.
	
	\bibitem{algogt}
	Noam Nisan, et al., (2007). "Algorithmic game theory," Cambridge: Cambridge University Press.
	
	\bibitem{gg}
	Michael Kearns, Michael L. Littman, Satinder Singh, (2001). "Graphical models for game theory," Proceedings of the Seventeenth conference on Uncertainty in artificial intelligence.
	
	\bibitem{poly}
	Joseph T. Howson, Jr.,  (1972). "Equilibria of polymatrix games." Management Science 18.5-part-1: 312-318.
	
	\bibitem{sparse}
	Xi Chen, Xiaotie Deng, Shang-Hua Teng, (2006). "Sparse games are hard." International Workshop on Internet and Network Economics. Springer Berlin Heidelberg.
	
	\bibitem{nash}
	John Nash, (1951). "Non-cooperative games." Annals of mathematics: 286-295.
	
	\bibitem{cg}	
	Robert W. Rosenthal, (1973). "A class of games possessing pure-strategy Nash equilibria." International Journal of Game Theory 2.1: 65-67.
	
	\bibitem{leg}	
	Kevin Leyton-Brown, Moshe Tennenholtz, (2005). "Local-effect games." Dagstuhl Seminar Proceedings. Schloss Dagstuhl-Leibniz-Zentrum für Informatik.
	
	\bibitem{ppad}
	Constantinos Daskalakis, Paul W. Goldberg, Christos H. Papadimitriou, (2009). "The complexity of computing a Nash equilibrium." SIAM Journal on Computing 39.1: 195-259.
	
	\bibitem{lemke}	
	Carlton E. Lemke, Joseph T. Howson, Jr., (1964). "Equilibrium points of bimatrix games." Journal of the Society for Industrial and Applied Mathematics 12.2: 413-423.
	
	\bibitem{wilson}	
	Robert Wilson, (1971). "Computing equilibria of n-person games." SIAM Journal on Applied Mathematics 21.1: 80-87.
	
	\bibitem{datta}	
	Ruchira S. Datta, (2003). "Using computer algebra to find Nash equilibria." Proceedings of the 2003 international symposium on Symbolic and algebraic computation. ACM.
	
	\bibitem{abst}	
	Andrew Gilpin, Tuomas Sandholm, (2007). "Lossless abstraction of imperfect information games." Journal of the ACM (JACM) 54.5: 25.
	
	\bibitem{seq}
	Bernhard Von Stengel, (1996). "Efficient computation of behavior strategies." Games and Economic Behavior 14.2: 220-246.
	
	\bibitem{branis}
	Branislav Bošanský, et al., (2015). "Combining compact representation and incremental generation in large games with sequential strategies." Proceedings of the Twenty-Ninth AAAI Conference on Artificial Intelligence. AAAI Press.
	
	\bibitem{mdp}
	Albert Xin Jiang, et al., (2013). "Game-theoretic randomization for security patrolling with dynamic execution uncertainty." Proceedings of the 2013 international conference on Autonomous agents and multi-agent systems. International Foundation for Autonomous Agents and Multiagent Systems.

\end{thebibliography}

\end{document}